\documentclass[preprint,12pt]{elsarticle}

\usepackage{graphicx}
\usepackage{amsmath}
\usepackage{amssymb}
\usepackage{amsthm}

\usepackage{url}
\usepackage{bbm}
\usepackage{accents}
\usepackage{booktabs}

\newtheorem{lemma}{Lemma}

\newtheorem{claim}{Claim}

\newtheorem*{theorem*}{Theorem}

\begin{document}

\begin{frontmatter}

\title{A model of discrete choice based on reinforcement learning under short-term memory}

\author{Misha Perepelitsa}

\date{\today}
\address{
misha@math.uh.edu\\
Department of Mathematics\\
University of Houston\\
4800 Calhoun Rd. \\
Houston, TX.}

\begin{abstract}

A family of models of individual  discrete choice are constructed by means of statistical averaging of choices made by a subject in a reinforcement learning process, where the subject has short, k-term memory span.  
The choice probabilities in these models combine in a non-trivial, non-linear way the initial learning bias and the experience gained through learning.  The properties of such models are discussed and, in particular, it is shown  that probabilities deviate from  Luce's Choice Axiom, even if the initial bias adheres to it.  Moreover, we shown that  the latter property is recovered as the memory span becomes  large.   

Two applications in utility theory are considered.  In the first, we use the discrete choice model to generate  binary preference relation on simple lotteries. We show that the preferences  violate transitivity and  independence axioms of expected utility theory. Furthermore, we establish the dependence of the preferences on frames, with risk aversion for gains, and risk seeking for losses. Based on these findings  we propose next a parametric model of choice based on the probability maximization principle, as a model for deviations from expected utility principle. To illustrate the approach  we apply it to the classical problem of demand for insurance.

\end{abstract}


\end{frontmatter}

\begin{section}{Introduction}
 
The problem of choice is one of the fundamental problems in psychology, economics and behavioral biology.  The second half of the last century saw a rapid growth of theoretical works on this subject as well as increasing amount of  experimental data. 
 
In economics, the field  was dominated by expected utility theory (EUT), its critique based on the experimental evidence and its ramifications.
 EUT was put forward by Von Neumann and Morgenstern (1947) as a mathematical formalization of what one can call rational preferences  between contingent prospects. EUT is an axiomatic theory that starts out with postulates about  preferences among prospects. The postulates are  completeness,  transitivity,  continuity and independence (substitution) axioms. The theory derives a utility function $u$  which  assigns values to payoffs, and a random prospect $X$ is ranked according to  its expectation $\mathbb{E}[u(X)].$ 
 
Psychology differs from economics in the approach to the choice behavior  by assuming  a more general description of choices as being probabilistic and dependent on the set of alternatives offered to a subject.  
 The axiomatic treatment for discrete choice was undertaken in a seminal work by Luce (1959), who introduced a choice axiom (Luce's choice axiom) that postulates how the probability to select an alternative from one set is related to the probability to select this alternative from a larger set.  Luce's theory establishes existence of a value function $v$ on a finite set of alternatives $T$  such that the probability to select $i$ from set $S\subset T$ equals
\begin{equation}
\label{scale}
P_S(i) {}={} \frac{v(i)}{\sum_{j\in S}v(j)},\quad i\in S.
\end{equation}
Psychologically interpreted,  value function (scale ratio)  $v(i)$ is a subject's response strength for alternative $i.$ Under Luce's choice axiom, choice probabilities verify  the principle of independence of irrelevant alternatives of Arrow (1951), and thus, when the latter is normative,  becomes a reasonable assumption.



Utility function and ratio scale function provide convenient tools for analysis of decision making. Empirically, the values of these  functions can obtained by comparing pairs of alternatives.

While the first three axioms of  EUT are generally accepted, the last one,
independence axiom,  drew a significant amount of critique from the experimentalists, starting with Allais (1953). Over the years several alternative utility theories were proposed  that provide some variants of the expected utility without the independence axiom or with its weaker  version. Among them, the generalized expected utility of Machina (1982), weighted utility theory of Chew and MacCrimmon (1979),
the regret theory developed independently by   Bell (1982), Fishburn (1982),  Loomis and Sugden (1982), rank dependent utility theory of Quiggin (1982, 1993), and the dual utility theory by Yaari (1987). Kahneman and Tversky (1979, 1984, 1992) based on their experimental findings, introduced  framing effects, value functions, and  probability weights into the analysis  and incorporated them into the prospect theory that was later developed, using the
approach of the rank dependent utility theory, into the cumulative prospect theory.

In the theory of discrete choice Luce's choice axiom (LCA) is not a universal imperative either and there are situations where it  does not apply, as in the example provided by Debreu (1960).  This example was further developed by Tversky (1972) who attributed it to the similarity effect and proposed aspect theory as a refinement of the probabilistic decision making.  More detailed discussion of validity of LCA can be found in Luce (1977), or in a more recent review of Pleskac (2013).

Let us now return to the work of Luce (1959) and mention its another major  contribution, this time to the field of reinforcement learning theories. Learning theories concern with subjects building  their choice probabilities through experience, by
adapting their responses accordingly to received stimuli. Following Bush and Mosteller (1951, 1955) the learning models were typically formulated 
in terms of the choice probabilities to select option $i$ at time $n:\, P^n(i),$ which is determined as function of the probabilities from the last period and an outcome of some random event conditioned on the last choice.

Luce (1957), citing  works of Thurstone (1930) and Gulliksen (1953), argued that learning must be formulated in terms of the strength of response to each alternative, with the choice probabilities being dependent variables of responses, for example, through the relation \eqref{scale}. Several learning models of this types were proposed typically with a linear law between response for alternative $i$ at times $n$ and $n-1:$  
\begin{equation}
\label{MP}
U^n(i){}={}\mu U^{n-1}(i) + (1-\mu)u^{n}(i),
\end{equation}
 where $u^{n}(i)$ is the response to the stimuli from the prior selection. This approach has been widely accepted by scientific community, and used in such fields as mathematical biology, game theory, and  engineering, see for example Harley (1981), Roth and Erev (1995, 1998), Fudenberg and Levine (1998), Sutton and Barto (1998).

There seems to be an unanimous agreement about general principles of reinforcement learning and, naturally, one can turn to it as the tool in constructing models of individual choice. In contrast with axiomatic approaches, one starts with a model of learning. Its parameters should be found experimentally, but once fixed, the model serves as a ``decision automaton'' that generates a series of choices, or in mathematical language, a stochastic process on a suitable state space. The process is analyzed for convergence to some stationary process, as the number of learning periods increases.  In the stationary process the choice probabilities might settle at certain values that  change little  as more and more choices are being made, but we will not force this assumption.  In particular, the stationary process  can be consistent with the case when a  repeated series of positive experiences with a particular alternative increases the probability to select this alternative in a non-negligible way. What is needed for the theory is existence of well-defined  statistical averages for the choice probabilities. The latter  are computed and recorded as the probabilities for a discrete choice model.  The model can be effectively described as an  ``expected probability'' choice model.  However, except for some special cases that we will mention later, it is not an LCA-type model  \eqref{scale}, nor any expected utility-type model.
 
The properties of such models is the focus of this work.  To that end we proceed, first, with the mathematical framework. In section \ref{sec1} we derive exact formulas for the choice probabilities for a  finite set of alternatives, showing by this that our approach is computationally feasible. We introduce a  class of $k$--term learning models, where in the process of learning a subject accounts only for responses to the last $k$ stimuli obtained for his/her actions. For example, with $k=1,$ the response strength to alternative $i$ at period $n$ is given by 
\begin{equation}
\label{mod1}
U^n(i) {}={}U^0(i) + u^n(i),\quad i\in S,
\end{equation}
where $U^0(i)$ is a learning prior (bias) for alternative $i,$ and $u^n(i)$ is either zero, if alternative $i$ was not selected at period $n-1,$ or, if it was, is the response to the received (random) stimulus for that alternative. Notice that the contribution of the  learning prior $U^0(i)$ is not diminishing  in the course of learning (as $n$ increases) and it will enter into the formulas of  asymptotic  probabilities. Another important assumption is implicitly included in \eqref{mod1}. The initial priors $U^0(i)$ do not depend on the subset $S$ of alternatives offered to the subject. That is, at the beginning of the experiment the subject has choice probabilities vitrifying Luce's choice axiom. Modified by learning, they enter into an asymptotic, expected probability model, which in general, loses that property.

In section \ref{sec2}  we turn to applications and our choice here will be  on deterministic binary preferences, since they take a prominent place in economics. We consider an individual presented with a set $T$ of  $q$ alternatives. For every pair of alternatives $S=\{x,y\}\subseteq T,$ the individual constructs choice probabilities for $x$ and $y$ from set $S,$ according to a learning process described above, where for simplicity we assume that the individual has very short, one period, memory span ($k=1$).    The binary preference  can be derived from a probabilistic choice model in many different ways, most obvious being the trace relation which defines $x\succ y$ iff the probability to chose $x$ from set $\{x,y\}$ is greater than 1/2. Depending on the parameters of the model,  we observe  a wide range of behaviors. There are some extreme cases when the binary preferences are according to EUT, while, generically, the preferences are not transitive and violate independence axiom. Intransitivity, in particular, implies that the choice probabilities, from which the binary preferences were derived, violate Luce's choice axiom.  There is more to it however, as we show that the binary preferences are characterized by a ``framing effect'', with risk averse preferences for gains and risk seeking  for losses, similar to the preferences in prospect theory of Kahneman and Tversky (1979, 1984). The binary preferences also shown to detect persistently better alternatives, by adhering to first order stochastic dominance. 

Intransitivity of the preferences and  violations of  independence axiom are two phenomena that typically enter any set of empirical data. The fact that they are revealed in our choice model, combined with the fact that the model  is founded  on the  behavioral   principles, warrants the interest in the experimental verification of the model.  This work however is limited only to  presentation of the model and its properties, not to establishing its empirical validity.  In section \ref{sec3} we introduce a parametric variant of the expected probability model with a partial motivation  to facilitate the task of performing  statistical tests.

The material of sections \ref{sec1} and \ref{sec2}  serves as a motivation for  a model for ``deviations from the EU principle", which we describe in section \ref{sec3}. 
The model is best described as a mediator between two expected utility principles.  
One is   based on maximization of response $\mathbb{E}[u(X)].$ The other is  based on the minimization of   
\[
\mathbb{E}\left[\frac{1}{1+ e^{u(X)}}\right],
\]
the quantity related to the expected probability. One can notice the dependence of this type of preferences on a  ``frame" through the shift of scales from $u$ to $u+u_0,$ which can change the ordering of preferences.   Interestingly enough it also establishes higher risk aversion for gains, and risk seeking for losses, even if $u$ itself is risk averse (concave), the phenomenon that we have mentioned earlier.
A generic case, described by formulas \eqref{deu:1}--\eqref{max_prob} on page \pageref{deu:1}, combines two types of EU principles into one. This case, however, is not EU-type principle any longer.  At the end of section \ref{sec3} we apply the model of ``deviations from EUT" to a classical problem of determining  demand for insurance.

\end{section}

\begin{section}{K-period reinforcement learning choice model}
\label{sec1}

\subsection{Response}
All reinforcement learning models have three ingredients in common. The reinforcement schedule, the response, as a function of the stimulus, and the choice probabilities depending on the  response strength. We will follow the approach of Luce (1959), according to which a subject has a mental record of
responses to each choice alternative and updates them according to the realized reinforcement (stimulus) for a corresponding alternative. Then, the subject implements choice through a subject-specific  function that selects an alternative with a probability proportional to the response strength for this choice. 

Consider a succession of experiments in which a subject is offered a stimulus for an alternative he chooses.
Let $r_i^n$ be the reinforcement given to the subject at $n^{th}$ trail,  when alternative $i$ was chosen last. We assume that each $r_i^n$ is sampled from a random variable $R_i,$ independently from other alternatives and  independently from one period to another. Reinforcement is measured in the experiment-dependent units such as dollars, carrots, intensities of light signals, etc.
 
Let $U_i^n$ be the total response strength for alternative $i$  that subject has for this alternative at the end of trial $n.$ It expresses the cumulative effect of the past reinforcements for alternative $i$ on the subject's attitude toward this alternative on some internal scales.

{\it Reinforcement-to-Response law:}  we will assume that the total response strength is additive over the incremental response strengths from each reinforcement. That is, there is a subject specific function $u=u(r)$ that maps the reinforcement value $r$ into the subject's response scales such that 
\begin{equation}
\label{RRL}
U_i^n{}={}U^0_i {}+{}\sum_{j\in[0..k-1]} u(r_i^{n-j})/N(n,i),
\end{equation}
where we agreed that if alternative $i$ is not selected during period $j,$ then $u(r_i^j)=0.$ $N(n,i)$ is the count of the number of times during the last $k$ trials when alternative $i$ was selected.  A partial justification for the short-memory model are the findings of Kaheman and Trversky (1979) in the context of utility theory. They report that choices are governed by the increments in the subject's wealth rather than the total accumulated wealth. 
The law \eqref{RRL} expresses the balance of contributions of the default response $U_i^0$ and new experience to the total response strength, when the former has non-diminishing contribution. \eqref{RRL} can be thought as a balance between  the instinct and the experience, or the balance between things learned in childhood, or bias,  and the presents reinforcement.

{\it Response-to-Probability law:} choice probability $P_i^n$ for alternative $i$ is determined through the relation
\begin{equation}
\label{def:Prob_i}
P_i^n{}={}\frac{\Phi(U_i^n)}{\sum_{j=1}^n \Phi(U_j^n)},
\end{equation}
where $\Phi=\Phi(U)$ is a re-scaling of the subject's response values $U$ into the range of strictly positive numbers. The function is selected to be non-decreasing to preserve the ordering of responses.

\subsection{Statistical count}
Consider now an experiment in which a subject is responding according to the rule just described. An outside observer with a capacity for statistical computations will notice that the proportion of  times each alternative $i$ is selected becomes fixated with time and settles at certain positive number $\bar{P}_T(i).$ Appendix section gives the formal mathematical account of this type of asymptotic behavior together with the formulas for probabilities $\left\{\bar{P}_T(i)\right\}.$ When the experiment is repeated with $T$ replaced by any subset $S\subset T,$ it leads to probabilities $\left\{\bar{P}_S(i)\right\}$ for alternatives to be selected from $S.$ In this way we obtain a family of probability measures 
\begin{equation}
\label{def:P}
\bar{P}_S(R){}={}\frac{\sum_{j\in R}\bar{P}_S(j)}{\sum_{j\in S} \bar{P}_S(j)},\quad R{}\subset{}S,
\end{equation} on subsets of $T.$ 
We will refer to this set of probabilities as RL(k) choice model. The properties of \eqref{def:P} are the main interest of the paper.  Their significance, hypothetically, arises from the fact that the subject  may learn the statistics of self behavior and use it next time the choice is to be made..  Or, he/she might run a quick mental ``simulation" of the learning process and come up with the same choice probabilities even when presented with choice once.  Better yet, with the memory span $k=1,$ when  the response is based only on the last reinforcement,  the subject  can ``generalize"  other people' one-time experiences about the alternatives as his/her own, in arriving at the probabilities $\bar{P}_T.$

These considerations warrant the interest in theoretical analysis of the family of probabilities \eqref{def:P}.  We start with few simple observations. Consider a special case when  an experiment provides no reinforcement whatsoever. Then,  the choice probabilities 
\[
\bar{P}_S(R) = \frac{\sum_{j\in R}\Phi(U_j^0)}{\sum_{j\in S}\Phi(U_j^0)},
\]
which verify Luce's Choice Axiom. Another extreme case occurs when the memory span becomes increasingly large.
\subsection{Long-term-memory limit}
Consider model RL(k) with increasingly large values of $k.$
We will proceed informally by noticing that the averages of responses to reinforcement in the last $k$ trails in \eqref{RRL} converge by the law of large numbers to the mean:
\[
\sum_{j\in[0..k-1]} u(r_i^{n-j})/N(n,i) \rightarrow \mathbb{E}[u(R_i)],\quad i=1..q.
\]
Thus for large values of $k$ and $n,$ the total response strength for alternative $i,$ $U_i^n\approx \bar{U}_i{}={}U_i^0{}+{}\mathbb{E}[u(R_i)],$ i.e. they change little from trial to trial. In the limit we obtain constant response levels and  the corresponding choice probabilities become LCA probabilities with the value scales $\Phi(\bar{U}_i),\,i=1..q.$ This is almost  the expected utility principle. It is in fact agrees with it if the priors $U_i^0$ are all equal, or if they are proportional to the corresponding expected response strength  $\mathbb{E}[u(R_i)].$










\end{section}


\begin{section}{Binary preference for lotteries}
\label{sec2}
The probabilistic description of choice is more general than algebraic (deterministic), and so, there are many different ways in which the latter can be derived from the former. Given a set of choice probabilities one can, for example, introduce a preference relation for alternatives by postulating some relation between the corresponding probabilities $P_{\{i,j\}}(i)$ and $P_{\{i,j\}}(j).$ One of such preferences is called trace relation which defines $i\succeq j$ if and only if $P_{\{i,j\}}(i){}\geq{}P_{\{i,j\}}(j).$ It was shown by Luce (1959), that if the family of probabilities $P_T$ verifies LCA, then trace relation is a weak order relation.

In this section we consider the trace relation for alternatives  that are monetary payoffs contingent on random events, with objectively known probabilities. We will assume that the positive scale function $\Phi(u){}={}exp(u/\beta),$ for some positive parameter $\beta,$ the response strength function $u(s)$ equals to the payoff $s$, and denote by  $R_i$ random payoff to alternative $i,$ and the initial bias $U^0_i{}={}\mathbb{E}[u(R_i)].$

The rationale for making such selections is the following.  Function $\Phi$  is the function from the logit probability model, originated from random utility theory, see Marschak (1960), and it is customary used as a scale function in the models of learning, see for example Fudenberg and Levine (1998).  The response $u$ is a linear function. We could have selected $u=ar+b,\,a>0,$ as a generic approximation of arbitrary non-decreasing function, but in order to simplify the exposition we restrict analysis to  the case $a=1,\,b=0.$ 
It reasonable to assume that the initial bias $U^0_i$ is some stochastic characteristic of the random payoff $R_i,$ known a priori, and having the units of ``utility'' $u(s).$ Thus,  the choice of $U^0_i$ as the expected ``utility''. In section \ref{sec3} we consider slightly more general model. Lastly, we restrict the analysis only to the learning models with the shortest  memory span $k=1.$ 


\subsection{Intransitivity of binary preferences}
Let $X,Y$ be two lotteries. Let $\bar{P}_1,\,\bar{P}_2$ be the equilibrium probabilities constructed from  model  RL(1) with set $T$ of two alternatives. We say that $X\succ Y$ if and only if $\bar{P}_1>\bar{P}_2.$ Equivalently (see appendix), the preference is defined by the inequality
\begin{equation}
\label{X>Y}
\mathbb{E}\left[\frac{e^{\mathbb{E}[Y]/\beta}}{e^{\mathbb{E}[Y]/\beta}{}+{}e^{((\mathbb{E}[X]+X)/\beta)}}
\right]{}< {}
\mathbb{E}\left[
\frac{e^{\mathbb{E}[X]/\beta}}{e^{\mathbb{E}[X]/\beta}{}+{}e^{((\mathbb{E}[Y]+Y)/\beta)}}
\right].
\end{equation}
We define relation  $X\succeq Y$ when the inequality in \eqref{X>Y} is not strict, and $X\sim Y$ when the expectations are equal. Notice that if the alternatives are restricted so that they have the same expected payoffs $\mathbb{E}[X]{}={}const.,$ then binary preference $\succ$ corresponds to EU principle with  ``utility" $ \tilde{u}(s)=(1+ e^{s/\beta})^{-1}.$ We will come back to this property in section \ref{sec3}.

Our first finding is that relation $\succ$ is not transitive. The proof of this is postponed to the appendix as it based on some technical manipulations with the integrals.

The intransitivity of the trace relation in its turn implies that choice probabilities $\bar{P}_S,$ violate Luce's choice axiom, see Luce (1959).

Intransitivity of preferences is a severe obstacle in constructing any choice theory which limits their usefulness. However, RL(1) model  provides much more information than just the binary preferences. In fact,  any finite number of lotteries can be evaluated and and  unambiguous rankings can be constructed.  We will exploit this property in full in section \ref{sec3}.

\subsection{Independence axiom and Allais experiment}
The independence axiom for preferences states that if lottery $L_1\succ L_2$
and $L$ is any other lottery, then a complex lottery in which $L_1$ is mixed with $L$ in some proportion dominates $L_2$ mixed with $L$ in the same proportion. This is the axiom for expected utility that drew the earliest critique to Neumann-Morgenstein theory, in particularly by Allais (1953).

The axiom requires preference to be linear in the probability distribution. 
Quick look at the formula \eqref{X>Y} reveals that the binary preference relation depends in a non-linear way on the distribution 
of the random variables. It is not surprising that such preferences  violate the independence axiom of the expected utility theory. To illustrate this property we consider an Allais-type experiment. In what follows we say that lottery $L(a,b:p)$ pays $\$a$ with probability $p$ and $\$b$ with probability $1-p.$ We shorten notation to $L(c:1)$ to describe lottery with 100\% of $\$c$ payoff. 

We compare lottery $L(c:1)$ and $L(1,0:x)$ according to \eqref{X>Y}. The result is represented graphically in figure \ref{fig:1} on page \pageref{fig:1}. The red line divides the unit square into two parts: below the line, the certain bet $L(c:1)\succ L(1,0:x),$ above, $L(1,0:x)\succ L(c:1),$ and the line itself is the indifference curve (certainty equivalent curve). 

Similar to Allais experiment, we mix each lottery  with  80\% chance of lottery $L(0:1)$ which pays \$0 for sure. The new lotteries to compare are $L(c,0:0.2)$ and $L(1,0:0.2x),$ and the indifference curve is marked blue in figure  \ref{fig:1}. The independence axiom requires the same preference between the new lotteries as between the original. For the binary preferences in question, however, for all pairs $(c,x)$ between  two curves in figure \ref{fig:1} preferences are reversed.

\subsection{Losses and Gains}
\label{losses}

Consider now the situation when the wealth increments are framed as losses (negative) or gains (positive).  In this setting we are going to look at the problem of determining the certainty equivalent for probabilistic lotteries involving only positive or only negative payoffs. 

First, we consider lottery $L(1,0:x)$ that pays \$0 with probability $1-x$ and \$1 with probability $x.$  For  such lottery we find its certainty equivalent \$c according to \eqref{X>Y}. Figure \ref{fig:2} on page \pageref{fig:2} shows (blue line) the corresponding indifference curve, in the first quadrant.  Then, we consider lotteries $L(-1,0:x)$ that pay \$0 with probability $1-x$ and $-1$ with probability $x.$ The red line in the third quadrant is the indifference curve for losses.  For comparison, we draw on the same figure the expected payoff curve
$\mathbb{E}[L(1,0:x)],$ which is the same the lottery weight $x.$

As seen from the figure, the indifference curve for the binary preferences  lies above expected utility in the region of gains, and below it, for losses.  Thus, an agent is more risk averse, compared to the expected utility preferences, in gains, and more risk tolerant  in losses. This example is a generic fact as proved in the next
\begin{lemma} Let $X$ be a random payoff of a lottery.  Then, if $X\geq 0$ ($X\leq 0$), the certainty equivalent of $X$ is less (greater) than the expected utility $\mathbb{E}[X].$ 
\end{lemma}
\begin{proof}
The certainty equivalent $c$ for lottery $X$ in the binary preferences is defined by the equation
\[
\mathbb{E}\left[ \frac{e^{c/\beta}}{e^{c/\beta}+e^{(\mathbb{E}[X] +X)/\beta}}\right]{}={}\frac{e^{\mathbb{E}[X]/\beta}}{e^{\mathbb{E}[X]/\beta} + e^{2c/\beta}}.
\]
Consider first the gains: $X\geq0.$ Let $p(c)$ denote the left-hand side of the equation, and $q(c)$ the right-hand side. We have: $p'(c)>0,$ and  $q'(c)<0.$  Moreover,
\[
p(\mathbb{E}[X]){}={}\mathbb{E}\left[ \frac{1}{1+e^{X/\beta}}\right]\geq \frac{1}{1+e^{\mathbb{E}[X]/\beta}}=q(\mathbb{E}[X]),
\]
because function $(1+e^{x/\beta})^{-1}$ is convex for $x\geq0.$ Thus, the point of intersection, $c,$ of graphs of $p$ and $q$ is not greater  than $\mathbb{E}[X].$ 

The other case is proved analogously, using the fact that $(1+e^{x/\beta})^{-1}$ is concave for $x\leq 0.$
\end{proof}

It should be noted that the situation in figure \ref{fig:2} applies only to two-valued lotteries described there. Since \eqref{X>Y} is not an expected utility no conclusions about other types of lotteries can be drawn from the graph of the certainty equivalents of such lotteries. In particular, the type of convexity of the certainty equivalent curve can not be used to characterize \eqref{X>Y} preferences in the relation to the attitude toward risk.

\subsection{First order stochastic dominance}
We will show in this section that an individual using  binary preferences \eqref{X>Y} can detect the prospect with persistently better payoffs than the other. This concept is formalized as first order stochastic dominance. 

Consider two lotteries $X$ and $Y.$ It is said that  $X$ stochastically dominates $Y,$ if for any  non-decreasing function $u$ and any non-increasing function $v,$  
\[
\mathbb{E}[u(Y)] \leq \mathbb{E}[u(X)],\quad \mathbb{E}[v(X)] \leq \mathbb{E}[v(Y)].
\]
\begin{lemma} If $X$ stochastically dominates $Y,$ then $X \succeq Y$ in the binary preferences \eqref{X>Y}.
\end{lemma}
\begin{proof}
Recall that $X$ is better than $Y$ if 
\[
\mathbb{E}\left[ \frac{e^{\mathbb{E}[Y]/\beta} }{e^{\mathbb{E}[Y]/\beta} + e^{(\mathbb{E}[X]+X)/\beta}}\right] \leq \mathbb{E}\left[ \frac{e^{\mathbb{E}[X]/\beta}}{e^{\mathbb{E}[X]/\beta} + e^{(\mathbb{E}[Y]+Y)/\beta}}\right].
\]
This inequality can be proved by using the definition of the stochastic dominance:
\begin{multline*}
\mathbb{E}\left[ \frac{e^{\mathbb{E}[Y]/\beta} }{e^{\mathbb{E}[Y]/\beta} + e^{(\mathbb{E}[X]+X)/\beta}}\right]{}={}
\mathbb{E}\left[ \frac{1}{1 + e^{(\mathbb{E}[X]-\mathbb{E}[Y])/\beta)}e^{X/\beta}}\right]\\
{}\leq{} \mathbb{E}\left[ \frac{1}{1 + e^{(-\mathbb{E}[X]+\mathbb{E}[Y])/\beta}e^{X/\beta}}\right]
{}\leq{} \mathbb{E}\left[ \frac{1}{1 + e^{(-\mathbb{E}[X]+\mathbb{E}[Y])/\beta}e^{Y/\beta}}\right]\\
{}={}\mathbb{E}\left[ \frac{e^{\mathbb{E}[X]/\beta}}{e^{\mathbb{E}[X]/\beta} + e^{(\mathbb{E}[Y]+Y)/\beta}}\right].
\end{multline*}
\end{proof}

\end{section}

\section{A model for deviations from EU principle}
\label{sec3}
In this section we slightly generalize the arguments presented above to introduce a parametric model for deviations of choices from EUT.  Suppose that we are to choose among lotteries $X_1,..,X_q$ offering random monetary payoffs. Let $u$ be a subject utility function, and let the subject form the choice probabilities based on RL(1) model with learning priors $U^0_i{}={}\mathbb{E}[u(X_i)]$ and using $u(X_i)$ to measure his/her responses.  
We introduce parameter $\alpha>0$ to measure  deviations from EU into the total response function, so that 
\begin{equation}
\label{deu:1}
U_i{}={}  \mathbb{E}[u(X_i)]{}+{}\alpha u(X_i),\quad \alpha>0,
\end{equation}
if $i$ was selected last, and 
\[
U_i{}={}  \mathbb{E}[u(X_i)],
\]
otherwise.

With the positive scales function $\Phi(u){}={}e^{u/\beta},$
equilibrium choice probabilities $\{\bar{P}(i)\}_{i=1}^q$ are determined by formula \eqref{formula:mu} from appendix:
\begin{equation}
\label{big:formula}
\bar{P}(i){}={}Ke^{U^0_i/\beta}\left(\mathbb{E}\left[
\frac{K_0}{K_0{}+{}e^{(U^0_i+\alpha u(R_i)/\beta}-e^{U^0_i/\beta}}\right]\right)^{-1},\quad  i=1..q,
\end{equation}
and $K_0 = \sum_i e^{U^0_i/\beta},$ and $K$ is a positive constant.

The choice is lottery  $X_{i_0}$ which maximizes the probability:
\begin{equation}
\label{max_prob}
\bar{P}(i_0){}={}\mbox{max}\{\bar{P}(1),..,\bar{P}(q)\}.
\end{equation}

Formula \eqref{big:formula} has two parts.  First, take $\alpha=0.$ This corresponds to the choices made on the basis of EU principle according to $\mathbb{E}[u(R_i)].$ 
 Now let $\alpha,\beta\to\infty,$ and at the same time fix the ratio $\alpha/\beta=1.$ In this case the subject is not using priors, but only one-period experience. This is also described by another  EU principle based on the minimization of 
\begin{equation}
\label{limit:1}
\mathbb{E}\left[\frac{1}{1+ e^{u(R_i)}}\right].
\end{equation}
To illustrate the properties such choices, consider, for example,  that the subject uses logarithmic $\log(1+s)$ function as the response increments (reinforcements), but it is shifted so that the response is  psychologically ``framed" at some reference value $u_0,$ so that 
$u<u_0$ is considered a loss and $u>u_0$ as gains. Thus,
\begin{equation}
\label{u:1}
u(s) = \log(1+s) - u_0,
\end{equation}
and without loss of generality we  assume that $u_0=\log(1+ s_0),$ for some positive $s_0.$ We will consider simple lotteries with payoffs in the interval $[0,2s_0]$ we determine how the subjects values them by the certainty equivalent. First we consider a lottery in the ``gains'' region of $[s_0,2s_0].$ The lottery pays $\$2s_0$ with probability $p$ and $\$s_0$ with probability $1-p.$ We parametrize $p$ by variable $x$ expressed in the units of utility:
\[
p{}={}\frac{x-u(s_0)}{u(2s_0)-u(s_0)}, \quad x\in[u(s_0),u(2s_0)].
\]
For such lottery we determine its certainty equivalent $c$ using \eqref{limit:1}, \eqref{u:1}, and mark point  $(c,s)$ on the graph in figure \ref{fig:3} on page \pageref{fig:3}. For comparison we plot the certainty equivalents for the lotteries but using expected utility $\mathbb{E}[u(X)],$ where $u$ is from \eqref{u:1}. In the last case the certainty equivalent curve is simply the graph of $u(s).$

For the region of losses, we repeat the construction using lottery that pays $\$s_0$ with probability $p,$ and $\$0$ with $1-p,$ where 
\[
p{}={}\frac{x-u(0)}{u(s_0)-u(0)},\quad x\in[u(0),u(s_0)].
\]
Figure \ref{fig:3} shows  certainty equivalents in this case as well. The certainty equivalent curve is below the graph of utility for losses, and above that for gains.

In between these two extremes, decision-making based on \eqref{big:formula}--\eqref{max_prob} shows deviations from the EU principle, of the similar kind as was discussed in the previous sections, on an example of linear utility $u(s)=s.$


Now we apply choice model \eqref{deu:1}--\eqref{max_prob} to  the following classical problem.
\subsection{Demand for insurance}
Consider a person contemplating purchasing an insurance against a loss of $\$\Delta$ that might occur next year, with objectively known probability $p.$
The next year earnings will  be $\$y.$ An insurance company offers  protection against the loss with actuarially fair premium $\delta =p\Delta.$ We assume that the person can purchase any level of insurance $\$a\Delta$ for the price of $\$a\delta,$ with  $a\in[-1,2].$ We assuming here that the person can overprotect ($a>1),$ or can actually borrow  cash on the promise to return a part of the loss if it occurs ($a<0$). In this set up we're dealing with a simple lottery $Y_a$ described in table \ref{tab:1}.

\begin{table}
\centering
\begin{tabular}{@{}ll@{}}
\toprule
payoff & probability\\
\hline
$y-ap\Delta$                        & $p$ \\
$y-\Delta + (1-p)a\Delta$      & $1-p$ \\
\bottomrule
\end{tabular}
\caption{Values of lottery $Y_a.$}
\label{tab:1}
\end{table}

Consider model \eqref{deu:1} with $\alpha=0,$ with logarithmic utility $u(s):$
\[
u(s)=\log(4+c)-\log8.
\]
It is an EU principle of maximization of $\mathbb{E}[u(Y_a)].$  For a concave utility $u(s),$ the solution is always $a=1,$ for any level of income $y,$ probability of loss $p,$ and loss $\Delta>0.$

Consider model \eqref{deu:1} with $\alpha/\beta =0.4,$ $\alpha,\beta\to+\infty,$  and same $u(s).$ Notice that  income $Y_a<4$ is treated as a loss and $Y_a>4$ as a gain.  
 We consider here the levels of income $y\in[0,10],$ and loss $\Delta =2,$ comparable to the income. 
 The selection is now based on the minimization of a functional similar to \eqref{limit:1}. The solution is shown in figure \ref{fig:4} on page \pageref{fig:4}. The figure  shows the amount of insurance $a$ one buys, as a function of the level on income $y$ and the loss probability $p.$ Due to different risk attitude for losses and gains the decision depends on the values of $y$ and $p.$ Notice that among all possible values of $a\in[-1,2]$ only three are being selected: $a=-1,1$ and $2.$ That is, the choice undergoes phase transition in $(y,p)$ values.

To illustrate the selection by non-EU choice model, take $\alpha = 0.4,$ $\beta=1.$ The  model applies only to finite number of alternatives. In fact, it depends non-trivially on the number of alternatives through parameter $K_0$ in \eqref{big:formula}.

We will give a person the choice between levels of insurance from $a=-1$ to $a=2$ with increment of $0.5,$ totaling to $q=7$ choices. Figure \ref{fig:5} shows which one is selected depending on $y$ and $p.$ Notice that again that the choice is mostly between three values $a=-1,1,2,$ similar to that of the second case, when $y>4.$ The same values, somewhat symmetrically appear in the region $y<4,$ and for $y<1$ solution becomes $a=1,$ as in the risk averse case.  Thus the model shows deviations from EUT ($a=1$) in the middle section of the figure, around a point of framing, the width of this region depending on the parameters of the model.

\begin{section}{Appendix}

\subsection{Model RL(k)}

In this section we will give a proof that in the course of the learning according to the process described section \ref{sec1}, the choice probabilities settle at some  equilibrium values and provide the formulas for them. To simplify the presentation but not the generality we will assume that choices are made from the set of all $n$ alternatives.

The learning process can be described as a Markov chain on the finite state space $S^k$ of $k$ most recent alternatives that a subject has selected, i.e.,
\[
S^k{}=\left\{
(i_1,..,i_k)\,:\, i_j{}\in{}1..n
\right\}.
\]
We denote the stochastic process as $X^n{}={}(i^n_1,..,i^n_k).$ From the formula \eqref{RRL} we see that the state $X^{n+1}$ is completely determined from the current state $X^n,$  i.e. $\{X^n\}_{n=0}^{\infty}$ is a Markov chain.

We will proceed with the computation of the transition probabilities from the state $\bar{i}=(i_1,..,i_k)$ to the state $\bar{m}=(m_1,..,m_k),$ that we denote by $p(\bar{m}:\bar{i}).$ This probability is zero unless $m_2=i_1,$ $m_3=i_2,..,m_k=i_{k-1}.$ In the remaining  cases, according to \eqref{RRL} and \eqref{def:Prob_i}
for $j=1..n,$
\begin{equation}
\label{def:trans_p}
p(j,i_1,i_2,..,i_{k-1}:i_1,..,i_k){}={}
\mathbb{E}\left[
\frac{\Phi(U_j)}{\sum_{l=1}^n \Phi(U_l)}
\right]{}>{}0,
\end{equation}
where 
\[
U_l{}={}U^0_l {}+{}\sum_{l\in[1..k]} u(R_{i_l})/N(l,i_1,..,i_k,R_{i_1},..,R_{i_k}),
\]
and $N(l,i_1,..,i_k,r_{i_1},..,r_{i_k})$ is a random variable that counts the number of times alternative $l$ has been selected, given the last $k$ selections $i_1,..,i_k$ and the last $k$ reinforcements $r_{i_1},..,r_{i_k}.$
The expectation in \eqref{def:trans_p} is with respect to the joint distribution of independent random variables $(R_{i_1},..,R_{i_k}).$

The Markov  chain $\{X^n\}_{i=0}^{\infty}$ is irreducible and all states are ergodic, see the monograph of Feller (1957) for the theory Markov chains.  This implies that the distribution of $X^n$ converges to an invariant measure on $S^k$ that we denote my $\mu.$
The probability that alternative $i$ has been selected last is computed from $\mu$ by the formula
\[
\bar{P}_T(i){}={}\sum_{i_2,..,i_k\in[1..n]} \mu(i,i_2,..,i_k).
\]
For the learning model with $k=1$ computation of $\{\bar{P}_T\}$ is somewhat simplified. The details are presented below as this is the case of the principle interest of the paper. 

For $k=1,$ $S^1{}={}\{1,..,n\}$ is simply the set of alternatives, and  
the transition probabilities equal 
\[
p(j:i) {}={} \mathbb{E}\left[
\frac{\Phi(U_j)}{\sum_{l=1}^n \Phi(U_l)}
\right]{}>{}0,
\]
where 
\[
U_l{}={}\left\{
\begin{array}{ll}
U^0_i {}+{} u(R_{i}), & l=i,\\
U^0_l, & l\not=i.
\end{array}\right.
\]
The vector of equilibrium measure $\mu=(\mu(1),..,\mu(n))^t$ is determined from the linear system
\begin{equation}
\label{mu:gen}
\mu {}={}M\mu,
\end{equation}
where the elements of the transition matrix $M$ equal $M_{j,i}{}={}p(j:i).$ The $i^{th}$ equation in this system reads:
\begin{multline*}
\mu(i){}={}\sum_{k\not=i}\mu(k)\mathbb{E}\left[
\frac{\Phi(U^0_i)}{\sum_{l\not =k} \Phi(U^0_l){}+{}\Phi(U^0_k+u(R_k))}\right]\\{}+{}
\mu(i)\mathbb{E}\left[
\frac{\Phi(U^0_i+u(R_i))}{\sum_{l\not =i} \Phi(U^0_l){}+{}\Phi(U^0_i+u(R_i))}
\right].
\end{multline*}
It can be rearranged as
\begin{multline*}
\mu(i){}={}\sum_{k}\mu(k)\mathbb{E}\left[
\frac{\Phi(U^0_i)}{\sum_{l\not =k} \Phi(U^0_l){}+{}\Phi(U^0_k+u(R_k))}\right]{}+{}
\\{}+{}
\mu(i)\mathbb{E}\left[
\frac{\Phi(U^0_i+u(R_i))-\Phi(U^0_i)}{\sum_{l\not =i} \Phi(U^0_l){}+{}\Phi(U^0_i+u(R_i))}
\right],
\end{multline*}
or as
\begin{multline*}
\mu(i)(\Phi(U^0_i))^{-1}\mathbb{E}\left[
\frac{\sum_l \Phi(U^0_l) }{\sum_{l\not =i} \Phi(U^0_l){}+{}\Phi(U^0_i+u(R_i))}\right]\\
{}={}
\sum_{k}\mathbb{E}\left[
\frac{\mu(k)}{\sum_{l\not =k} \Phi(U^0_l){}+{}\Phi(U^0_k+u(R_k))}\right].
\end{multline*}
The right-hand side of this equation is independent of $i;$ we denote it by $K.$ Also we assign $K_0{}={}\sum_{l} \Phi(U^0_l).$ Then,
 the last equation provides a formula
\begin{equation}
\label{formula:mu}
 \mu(i){}={}K\Phi(U^0_i)\left(\mathbb{E}\left[
\frac{K_0}{K_0{}+{}\Phi(U^0_i+u(R_i))-\Phi(U^0_i)}\right]\right)^{-1},\quad  i=1..q.
\end{equation}
The set of equilibrium choice probabilities $\{\bar{P}_T\}$ is the same as measure $\mu.$
If only two alternatives are present, $T=\{1,2\},$ \eqref{mu:gen} can be reduced to a single equation, since $\mu(1)+\mu(2)=1:$
\begin{equation}
\label{binary}
\mu(1)\mathbb{E}\left[
\frac{\Phi(U^0_2)}{\Phi(U^0_2){}+{}\Phi(U^0_1+u(R_1))}
\right]{}={}
\mu(2)\mathbb{E}\left[
\frac{\Phi(U^0_1)}{\Phi(U^0_1){}+{}\Phi(U^0_2+u(R_2))}
\right],
\end{equation}
which we use in applications to compute ratio $\mu(1)/\mu(2).$

\subsection{Intransitivity}

\begin{lemma} For the binary preferences defined in \eqref{X>Y}, there are lotteries $X,Y,Z$ such that $X\succ Y,$ $Z\succ X$ but $Y\succ Z$.
\end{lemma}
\begin{proof}
We will show that there are $X,\,Y,\,Z$ such that $X\succeq Y,$ $Z\succeq X$ but $Y\succ Z,$ i.e. that $\succeq$ is not transitive. The lotteries $X,\, Y,\,Z$ can be suitably perturbed to show that $\succ$ is not transitive as well.

Let $X,Y$ be two lotteries such that $\mathbb{E}[X]=\mathbb{E}[Y]=0,$
\begin{equation}
\label{in:1}
\mathbb{E}\left[ \frac{1}{1+e^{X}}\right]{}={}\mathbb{E}\left[ \frac{1}{1+e^{Y}}\right],
\end{equation}
and 
\begin{equation}
\label{in:2}
\mathbb{E}\left[ \frac{e^X}{(1+e^{X})^2}\right]{}>{}\mathbb{E}\left[ \frac{e^Y}{(1+e^{Y})^2}\right]
\end{equation}
Let $\bar{Z}$ be a lottery with $\mathbb{E}[\bar{Z}]=0$ and 
\begin{equation}
\label{in:3}
\mathbb{E}\left[ \frac{1}{1+e^{\bar{Z}}}\right]{}>{}\mathbb{E}\left[ \frac{1}{1+e^{X}}\right].
\end{equation}
First we establish the following 
\begin{claim} For any number $z_0>0$ there is $z\in(0,z_0)$ and a lottery $Z$ such that $\mathbb{E}[Z]=z,$ and  
\begin{equation}
\label{in:4}
\mathbb{E}\left[ \frac{1}{1+e^{z+Z}}\right]{}={}\mathbb{E}\left[ \frac{1}{1+e^{-z+X}}\right].
\end{equation}
\end{claim}
\begin{proof}
Consider functions $f(z) = \mathbb{E}\left[ \frac{1}{1+e^{2z+\bar{Z}}}\right]$ and 
$g(z) = \mathbb{E}\left[ \frac{1}{1+e^{-z+X}}\right].$ From \eqref{in:3} we have $f(0)>g(0).$ Moreover $f(z)$ is monotone function with $f'(0)<0,$ and $g(z)$ is monotone with $g'(0)>0$ as can be verified by taking derivatives. If lottery $\bar{Z}$ is chosen sufficiently close to $X,$ graphs of $f(z)$ and $g(z)$ will  intersect at some point in the interval $(0,z_0).$ For such point, call it $z,$ and lottery $Z = z+\bar{Z},$
equation \eqref{in:4} holds.
\end{proof}
Condition \eqref{in:1} means that $X\succeq Y.$ Consider function $f(z) = \mathbb{E}\left[ \frac{1}{1+e^{-z+X}}\right]$ and 
$g(z) = \mathbb{E}\left[ \frac{1}{1+e^{-z+Y}}\right].$ Conditions \eqref{in:1},\eqref{in:2} imply that there is $z_0>0$ such that for all $s\in(0,z_0),$
\[
\mathbb{E}\left[ \frac{1}{1+e^{-s+X}}\right] > \mathbb{E}\left[ \frac{1}{1+e^{-s+Y}}\right].
\]
Given this $z_0$ and using the claim we find $z\in(0,z_0)$ and lottery $Z,$ such \eqref{in:4} holds, which implies that $Z\succeq X.$ On the other hand
\[
\mathbb{E}\left[ \frac{1}{1+e^{-z+X}}\right] > \mathbb{E}\left[ \frac{1}{1+e^{-z+Y}}\right].
\]
Thus, 
\[
\mathbb{E}\left[ \frac{1}{1+e^{z+Z}}\right]{}>{}\mathbb{E}\left[ \frac{1}{1+e^{-z+X}}\right].
\]
By formula \eqref{X>Y} it implies that $Y\succ Z$ (recall that $\mathbb{E}[Z]=z,\, \mathbb{E}[Y]=0$), while \eqref{in:1} establishes that $X\succeq Y.$

\end{proof}

\end{section}

\begin{figure}[t]
\centering
\includegraphics[width=13cm]{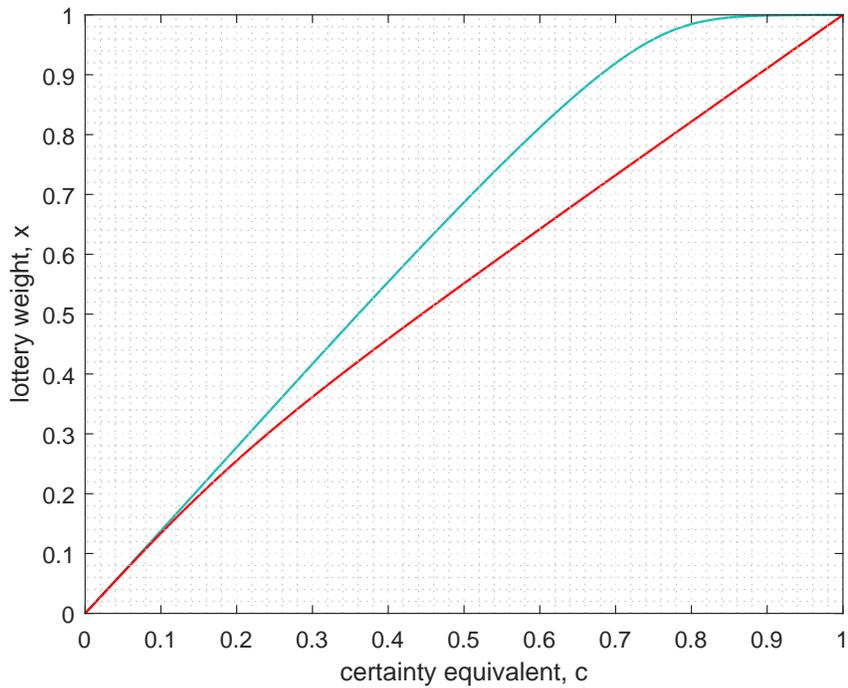}
\caption{Allais Paradox. Red line is the certainty equivalent for lottery $L(1,0:x).$   Blue  line represents  points $(c,x)$ for which lotteries $L(1,0:0.2x)$ are equivalent to $L(c,0:0.2).$ The region between two curves is set of points for which Allais paradox holds.
 RL(1) model has positive scale function $\Phi(u)=e^{u/0.1}$ and learning priors $U^0=\mathbb{E}[L(1,0:x)]=x.$ \label{fig:1}}

\end{figure}

\begin{figure}[t]
\centering
\includegraphics[width=13cm]{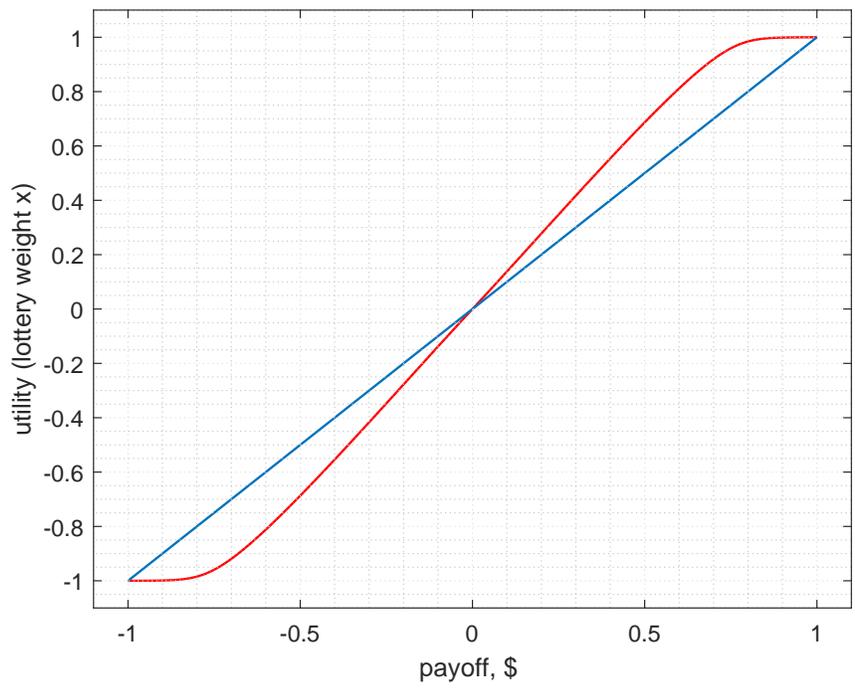}
\caption{Certainty equivalent curves for simple lotteries for losses and gains according to  EUT (blue) and the binary preferences from RL(1) model (red).  The function of positive scales $\Phi(u)=e^{u/.1}.$ \label{fig:2}}
\end{figure}

\begin{figure}[t]
\centering
\includegraphics[width=13cm]{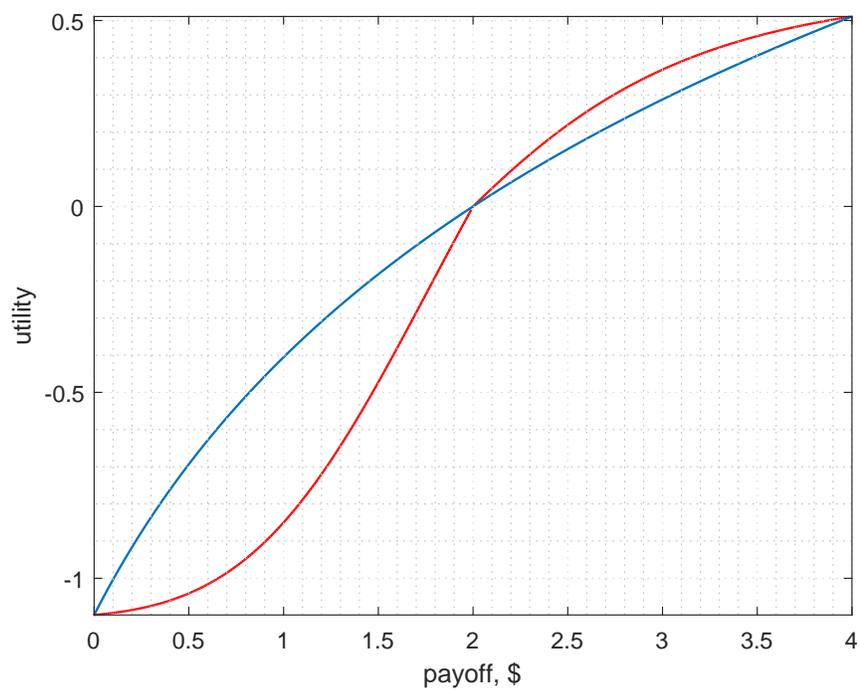}
\caption{Certainty equivalent curves for simple lotteries for losses and gains according to  EU (blue) and model \eqref{deu:1}-\eqref{max_prob} (red).  Utility $u(c) = \log(1+c)-\log3.$ Transition from losses to gains occurs at $c=2.$ The function of positive scales $\Phi(u)=e^{u/.2}.$ \label{fig:3}}
\end{figure}

\begin{figure}[t]
\centering
\includegraphics[width=13cm]{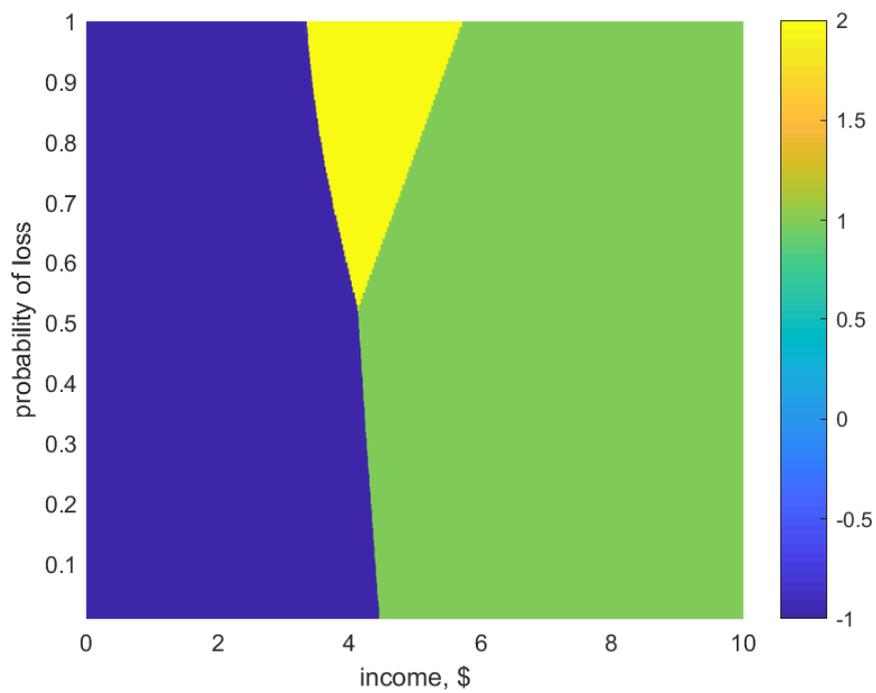}
\caption{Demand for insurance I. The figure shows fraction $a$ of purchased insurance (in color) for every pair of $(y,p)$ -- level of income and probability of loss, according to the model \eqref{deu:1}--\eqref{max_prob} with parameters $\alpha,\beta=+\infty,$ $\alpha/\beta{}={}0.4.$  \label{fig:4}}
\end{figure}

\begin{figure}[t]
\centering
\includegraphics[width=13cm]{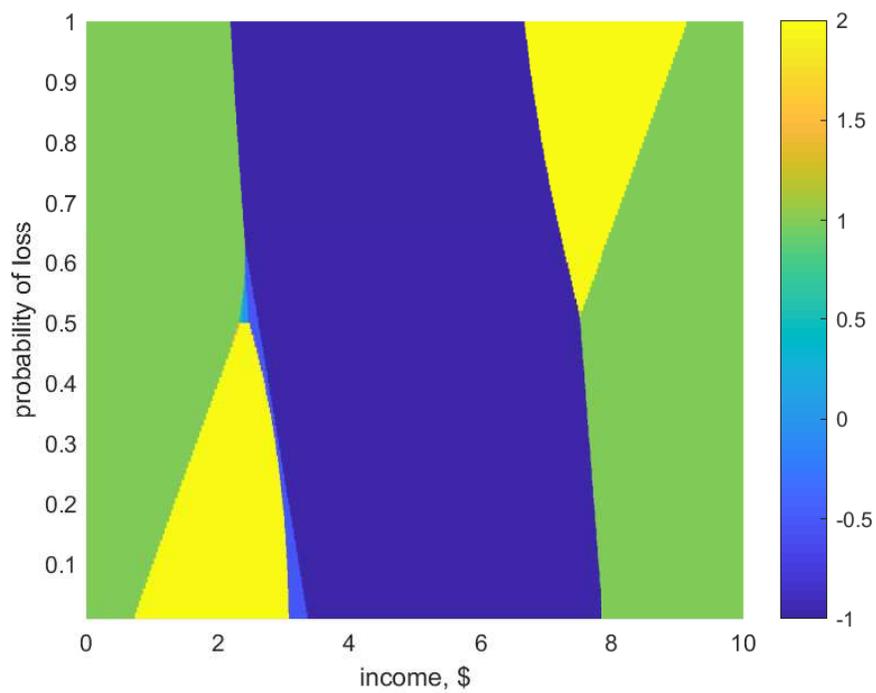}
\caption{ Demand for insurance II. The figure shows fraction $a$ of purchased insurance (in color) for every pair of $(y,p)$ -- level of income and probability of loss, according to the model \eqref{deu:1}--\eqref{max_prob} with parameters $\alpha=0.4,\beta=1.$ \label{fig:5}}
\end{figure}

\end{document}